\documentclass[reqno,12pt,letterpaper]{amsart}
\usepackage{amsmath,amssymb,amsthm,mathtools,amsfonts,cases}
\usepackage{graphicx}
\usepackage{amsxtra}
\usepackage[alphabetic,nobysame]{amsrefs}
\usepackage{slashed}
\usepackage{ulem}

\usepackage{bm}

\usepackage[top=1in,left=1in,right=1in,bottom=1in]{geometry}
\usepackage[usenames,dvipsnames]{color}
\usepackage[colorlinks=true,linkcolor=Red,citecolor=Green]{hyperref}

\usepackage{bbold}

\newcommand{\RR}{{\mathbb R}}
\newcommand{\NN}{{\mathbb N}}
\newcommand{\CC}{{\mathbb C}}

\newcommand{\ZZ}{{\mathbb Z}}

\newcommand{\MG}{{\mathcal G}}

\newcommand{\bfx}{{\bm{x}}}
\newcommand{\bfy}{{\bm{y}}}
\newcommand{\bfz}{{\bm{z}}}
\newcommand{\bfw}{{\bm{w}}}
\newcommand{\1}{{\mathbb{1}}}
\newcommand{\bfn}{{\textbf{n}}}
\newcommand{\Fr}{{\operatorname{Fr}}}

\DeclareMathOperator{\supp}{Supp}
\DeclareMathOperator{\Tr}{Tr}

\newtheorem{theorem}{Theorem}
\newtheorem{prop}{Proposition}

\newtheorem{lemma}{Lemma}[section]

\newtheorem{remark}{Remark}
\newtheorem{assumption}{Assumption}

\numberwithin{equation}{section}
\mathtoolsset{showonlyrefs}

\title{Topological edge spectrum along curved interfaces}
\author{Alexis Drouot}
\address[Alexis Drouot]{University of Washington, Seattle, USA.} 
\email{adrouot@uw.edu}
\author{Xiaowen Zhu}
\address[Xiaowen Zhu]{University of Washington, Seattle, USA.} 
\email{xiaowenz@uw.edu}

\begin{document}
\begin{abstract}
    We prove that that if the boundary of a topological insulator divides the plane in two regions containing arbitrarily large balls, then it acts as a conductor. Conversely, we show that topological insulators that fit within strips do not need to admit conducting boundary modes.
    \end{abstract}

\maketitle


\section{Introduction and main results}

\subsection{Introduction} Topological insulators are novel materials with striking properties. They are phases of matter insulating in their bulk (the Hamiltonian has a spectral gap), but that turn into conductors when truncated to half-spaces (the spectral gap fills). The resulting edge conductance is equal to the difference of the bulk topological invariants across the cut, a principle known as the bulk-edge correspondence \cites{KS02,EG02,EGS05,ASV13,GP13,GT18,ST19,B19,D21,LT22}. Here, we consider truncations of topological insulators in regions more sophisticated than half-spaces (for example, sectors or filled parabolas). We investigate how the shape of the resulting material affects the spectrum.

\begin{figure}[b]
  \centering
  \includegraphics[width=1\textwidth]{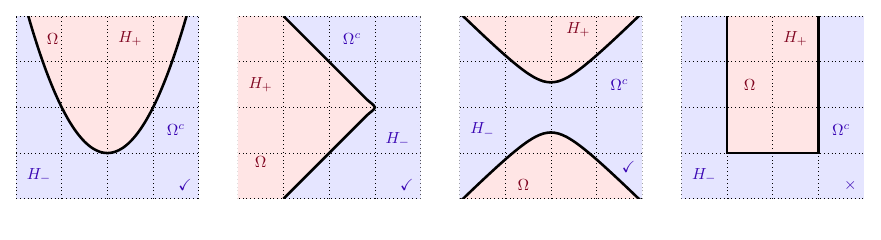}
  \caption{The first three truncations systematically give rise to edge spectrum. In the last case, we construct Haldane-type topological insulators $H_-, H_+$ such that $H_e$ is also an insulator.}
  \label{fig:1}
\end{figure}

The main operator in our analysis is approximately equal to
\begin{equation}
H_e := \begin{cases}
    H_+ \quad \text{ on } & \Omega,\\
    H_- \quad \text{ on } & \Omega^c = \ZZ^2 \setminus \Omega,
\end{cases}
\end{equation}
for two Hamiltonians $H_\pm$ on $\ell^2(\ZZ^2,\CC^d)$ with distinct bulk invariants within a joint spectral gap $\MG$ (potentially one of them representing the vacuum); and $\Omega$ a subset of $\ZZ^2$. We refer to \S\ref{sec-1.2} for precise definitions and assumptions. We ask for geometric conditions on $\Omega$ that guarantee that $H_e$ has spectrum filling $\MG$. 

Our main result, Theorem \ref{thm:1}, asserts that if both $\Omega$ and $\Omega^c$ contain arbitrarily large balls then $H_e$ has spectrum filling $\MG$ (referred to as edge spectrum). Hence $H_e$ behaves like a conductor near $\partial \Omega$. Examples of such domains $\Omega$ include half-planes, sectors, regions enclosed by hyperbolas, and so on; they exclude strips or half-strips -- see Figure \ref{fig:1}. In this last case, we actually show that there exist examples of distinct topological insulators $H_\pm$ such that $H_e$ remains an insulator -- Theorem \ref{thm:3}. Therefore, topological materials fitting in strips can violate the bulk-edge correspondence: boundaries or interfaces are not systematically conductors. This violation was suggested in a question of G.M. Graf in an online talk by G.C. Thiang \cite{GT20}.

Our experience of the world, is by nature, an approximation of the reality. Experiment samples -- here, $\Omega$ -- are always finite and spectral measurements are valid only up to an uncertainty $\delta > 0$. Hence, in practice $\Omega$ never contains arbitrarily large balls -- but neither can experiments assess with full certainty that spectral gaps completely fill: they can only measure emergence of a $\delta$-dense set of spectrum. Theorem \ref{thm:2} is a quantitative formulation of Theorem \ref{thm:1} that relates to these observations. It predicts that there exist constants $\alpha, R_0$ such that for all $R \geq R_0$, the following holds. Assume that both $\Omega, \Omega^c$ contain a ball of radius $R$. Then the spectrum of $H_e$ within $\MG$ is $C \ln(R)/R$-dense:
\begin{equation}
    \forall \lambda_* \in \MG,  \ \ \exists \lambda \in \Sigma(H_e), \ \ |\lambda-\lambda_*| \leq \dfrac{\alpha \ln(R)}{R}.
\end{equation}
This justifies why topological insulators truncated to sufficiently large balls appear in experiments conducting along their boundaries.

\subsection{Topological insulators and interface operators}\label{sec-1.2} We briefly review standard facts from condensed matter physics. Electronic propagation through a given material is described via a selfadjoint operator $H$ on a Hilbert space -- here $\ell^2(\ZZ^2,\CC^d)$. The spectrum $\Sigma(H)$ of $H$ characterizes the electronic nature of the material: $H$ is a conductor at energy $\lambda$ if and only if $\lambda \in \Sigma(H)$ and an insulator otherwise.
\textbf{\begin{center} In the rest of this paper, $\nu \in (0,1]$ is a fixed parameter.\end{center}}

We work here with short-range Hamiltonians: operators on $\ell^2(\ZZ^2,\CC^d)$ whose kernels satisfy
\begin{equation}\label{eq-1c}
   \forall \bfx,\bfy \in \ZZ^2, \qquad \big| H(\bfx,\bfy) \big| \leq \nu^{-1} e^{-2\nu |\bfx-\bfy|}, \qquad |\bfx-\bfy| := |x_1-y_1|+|x_2-y_2|. 
\end{equation}
Under \eqref{eq-1c}, one can define the bulk conductance of $H$ at an insulating energy. For $\lambda \notin \Sigma(H)$, let $P_\lambda(H) = \1_{(-\infty,\lambda)}(H)$ be the spectral projector below energy $\lambda$, and $\Lambda_1$ (respectively $\Lambda_2$) the indicator function of $\NN \times \ZZ$ (respectively $\ZZ \times \NN$). Then the operator $P_\lambda(H) \big[[P_\lambda(H),\Lambda_1],[P_\lambda(H),\Lambda_2]\big]$ is trace-class (see \cite{EGS05} and Remark \ref{rmk: interger_conductance} below) and 
\begin{equation}
    \sigma(H,\lambda) := -2\pi i \Tr \left(P_\lambda(H) \big[[P_\lambda(H),\Lambda_1],[P_\lambda(H),\Lambda_2]\big]\right)
\end{equation}
is well-defined. We comment that if $\MG$ is a subinterval of $\Sigma(H)^c$ (referred to below as a spectral gap), then 
\begin{equation}
    \lambda, \lambda' \in \MG \ \Rightarrow \ \sigma(H,\lambda) = \sigma(H,\lambda').
\end{equation}
Therefore, there is no ambiguity in using the notation $\sigma(H,\MG)$ for $\sigma(H,\lambda)$, $\lambda \in \MG$.  It represents the bulk conductance  for energies in $\MG$ \cite{ASS94}. Under the gap condition $\lambda \notin \Sigma(H)$, $\sigma(H,\lambda)$ is an integer that measures topological aspects of the Hamiltonian $H$.

In this work, we ask under which conditions interfaces between two topologically distinct insulators -- the bulk materials -- carry currents. We make the following assumption on the bulk components:

\begin{assumption}\label{a1} $H_\pm$ are two selfadjoint, short-range Hamiltonians on $\ell^2(\ZZ^2,\CC^d)$, with a common spectral gap $\MG$ (an interval contained in $\Sigma(H_+)^c \cap \Sigma(H_-)^c$) and distinct bulk conductances within~$\MG$:
\begin{equation}
    \sigma(H_+,\MG) \neq \sigma(H_-,\MG).
\end{equation}
\end{assumption}

Given a domain $\Omega \in \ZZ^2$, we define its boundary by 
\begin{equation}
    \partial \Omega := \{ \bfx \in \Omega, \ B_1(\bfx) \not\subset \Omega \} \cup \{ \bfx \in \Omega^c, \ B_1(\bfx) \not\subset \Omega^c \},
\end{equation}
where $B_r(\bfx)$ denotes the $\ell^1$-ball of radius $r$ centered at $\bfx$ in $\ZZ^2$. We will denote the distance from $\bfx$ to $\partial \Omega$ by $d(\bfx,\partial \Omega)$. We make the following assumptions on the interface operator:

\begin{assumption}\label{a2} $H_e$ is a selfadjoint, short-range Hamiltonian on $\ell^2(\ZZ^2,\CC^d)$ satisfying the kernel condition: \begin{equation}\label{eq-1a}
\forall \bfx, \bfy \in \ZZ^2, \qquad
    \big| E(\bfx,\bfy) \big| \leq \nu^{-1} e^{-2\nu d(\bfx,\partial \Omega)}, \qquad E :=  H_e - \1_\Omega H_+ \1_\Omega - \1_{\Omega^c} H_- \1_{\Omega^c}.
\end{equation}
\end{assumption}

The condition \eqref{eq-1a} means that $H_e$ is equal to $H_+$ on $\Omega$ and $H_-$ on $\Omega^c$, up to corrections decaying exponentially away from $\partial \Omega$.

\subsection{Main results}\label{sec-1.3} To formulate our main results, we will need the notion of filling radius for a subset $\Omega$ of $\ZZ^2$:\
\begin{equation}
    \Fr(\Omega) = \sup \{ r : \ \exists \ \bfx \in \ZZ^2, \ B_r(\bfx) \cap \ZZ^2 \subset \Omega \}.
\end{equation}
 It measures the size of the largest ball contained in $\Omega$: $\Fr(\Omega) \geq r$ if and only if $\Omega$ contains a ball of radius $r$. 

\begin{theorem}\label{thm:1} Let $H_\pm, H_e$ satisfying Assumptions \ref{a1} and \ref{a2}. If $\Fr(\Omega) = \Fr(\Omega^c) = \infty$, then $\MG \subset \Sigma(H_e)$.  
\end{theorem}

This means that if the boundary of a topological insulators divides the plane in two regions of infinite filling radius, then it is a conductor. Theorem \ref{thm:1} will actually follow from a more quantitative statement. Given $\delta > 0$, we say that a set  $S \subset \MG$ is $\delta$-dense within $\MG$ if either $|\MG| < \delta$
or if 
\begin{equation}
    \forall \lambda_* \in \MG, \ \exists \lambda \in S \text{ s.t. } |\lambda_*-\lambda| \leq \delta.
\end{equation}

\begin{theorem}\label{thm:2} There exist constants $\alpha_\nu > 0$ and $R_\nu > 0$, depending on $\nu$ only, such that the following holds. For $H_\pm, H_e$ satisfying Assumptions \ref{a1} and \ref{a2} and any $R \geq R_\nu$:
\begin{center}
    $\Fr(\Omega), \ \Fr(\Omega^c) \geq R \quad \Rightarrow \quad \Sigma(H_e) \cap \MG \ \ $  is \  $\alpha_\nu\frac{\ln R}{R}$--dense within $\MG$.
\end{center}
\end{theorem}

Theorem \ref{thm:2} has the following physical interpretation. Assume that $H_e$ represents the truncation of a topological insulator $H_+$ in the ball $B_R(0)$, and that we have a measurement procedure that can infer if an energy is within $\delta$ of $\Sigma(H_e)$. Theorem \ref{thm:2} asserts that if $\ln(R)/R \ll \delta$, then experiments measure that the spectral gap of $H_+$ closes when truncating it to  $B_R(0)$. 
This imperfect conclusion ($H_e$ actually has discrete spectrum when truncated to $B_R(0)$) is due to the limitation of the measuring procedure. 

Theorem \ref{thm:2} implies Theorem \ref{thm:1}: if $\Fr(\Omega) = \Fr(\Omega^c) = \infty$, then $\Fr(\Omega)$ and $\Fr(\Omega^c)$ are larger than $R$ for any $R$, so $\Sigma(H_e) \cap \MG$ is $\delta$-dense within $\MG$ for any $\delta > 0$. This means that $\Sigma(H_e) \cap \MG$ is actually dense within $\MG$; since it is a closed subset of $\MG$ we conclude that $\Sigma(H_e) \cap \MG = \MG$, equivalently $\MG \subset \Sigma(H_e)$. 

A natural question is whether the conclusion of Theorem \ref{thm:1} fails for unbounded sets $\Omega$ with finite filling radius. 

\begin{theorem}\label{thm:3} Fix $L > 0$ and $\Omega \subset [-L,L] \times \ZZ$. There exists operators $H_\pm$ satisfying Assumption \ref{a1} for a gap $\MG$ containing $0$; and $H_e$ satisfying Assumption \ref{a2} such that $0 \notin \Sigma(H_e)$.     
\end{theorem}

In other words, a topological insulator fitting in a strip does not have to be a conductor. This constitutes a violation of the bulk-edge correspondence: an interface between two distinct topological phases does not have to support edge states, for instance if the interface is the boundary of a half-strip.

\begin{figure}
    \centering
    \begin{minipage}{0.4\textwidth}
        \includegraphics[width=0.8\linewidth]{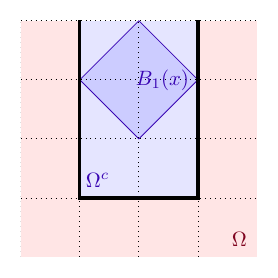}
    \end{minipage}
    \begin{minipage}{0.55\textwidth}
        \centering
        \caption{A region $\Omega$ that does not satisfy the assumption of Theorem \ref{thm:1}: the largest ball that fits in $\Omega^c$ has radius $1$, so $\Fr(\Omega^c) = 1 < \infty$. There actually exist operators $H_\pm, H_e$ satisfying Assumptions \ref{a1}, \ref{a2} such that $\MG \not \subset \Sigma(H_e)$.}
    \end{minipage}
\end{figure}

\subsection{Sketch of proof.} We explain here the main ideas leading to Theorem \ref{thm:1}. Strictly speaking, the paper will focus on quantitative forms of these ideas to obtain Theorem \ref{thm:2}, which (as explained above) implies Theorem \ref{thm:1}.

Let $H$ be a selfadjoint, short-range Hamiltonian on $\ell^2(\ZZ^2,\CC^d)$ and $\lambda \notin \Sigma(H)$. Our main argument is the observation that $\sigma(H,\lambda)$ can be computed from the sole knowledge of $H$ within a ball centered at any point $\bfn$ of $\ZZ^2$, as long as its radius $R$ is sufficiently large (depending on $\nu$, but independent of $\bfn$). While this may seem paradoxical, the global condition $\lambda \notin \Sigma(H)$ actually ensures that the result of the computation is independent of $\bfn$.

We proceed now by contradiction. Let $H_\pm, H_e$ satisfy the assumptions of Theorem \ref{thm:1}, and assume that $\lambda \in \MG \setminus \Sigma(H_e)$. We can then define $\sigma(H_e,\lambda)$, and, thanks to the above observations, compute it from sole knowledge of $H_e$ on balls of the form $B_R(\bfn)$. 

Pick now $\bfn$ so that $B_R(\bfn)$ lies deep in $\Omega$, i.e. in $\Omega$ and far from $\partial \Omega$; this is possible because under $\Fr(\Omega) = \infty$, $\Omega$ contains arbitrarily large balls. From Assumption \ref{a2}, $H_e$ is roughly $H_+$ there and we deduce $\sigma(H_e,\lambda) = \sigma(H_+,\lambda)$. Likewise, pick $\bfn$ so that $B_R(\bfn)$ lies deep in $\Omega^c$ and deduce $\sigma(H_e,\lambda) = \sigma(H_-,\lambda)$. This contradicts the assumption $\sigma(H_+,\lambda) \neq \sigma(H_-,\lambda)$. Therefore, $\lambda \in \Sigma(H_e)$.

We comment that this proof also applies to materials made off three or more topological insulators, with at least two of them with different bulk invariant filling regions with infinite filling radius. 

\subsection{Relation to existing results} The question of how the shape of the truncation affects the edge spectrum has considered before. The bulk-edge correspondence predicts the emergence of edge spectrum for half-space truncations: it gives the resulting interface conductance as a difference of Chern numbers \cites{KS02,EG02,EGS05,ASV13,GP13,B19,D21,LT22}.

In \cite{FGW00}, the authors focus on truncated quantum Hall Hamiltonians and derive a global analytic condition on $\Omega$ for the emergence of edge spectrum. They verify that this condition holds for regions with asymptotically flat boundary. This includes local perturbation of sectors. 

More recently the techniques have drifted to coarse geometry and K-theory. In \cites{T20,O22} the authors prove that magnetic Hamiltonians truncated to corners or sectors, and their local perturbations, have edge spectrum. The furthest-reaching work is due to Ludewig--Thiang \cite{LT22}. It predicts that truncations of general topological insulators to domains coarsely equivalent to half-planes admit edge spectrum. 

Our work derives a simple visual criterion for emergence of edge spectrum: both $\Omega$ and $\Omega^c$ contain arbitrary large balls (equivalently, $\Fr(\Omega) = \Fr(\Omega^c) = \infty$). It is not evident how our condition relates to those of \cites{FGW00,LT22}. Our proof returns to an analytic approach, which has the benefit to come up with a quantitative form of the result. This version explains  in what sense experimentalists observe edge spectrum in bounded samples. 

The shape of edge states matters in technological applications: they are the vectors of conduction along the edge. When the boundary is weakly curved -- which correspond to the adiabatic or semiclassical regime -- several works constructed edge states as wavepackets \cites{bal2021edge,bal2022semiclassical,D22,PPY22}. The assumptions in the present work are significantly weaker -- we only assume existence of a spectral gap -- but the result is also significantly weaker: we only prove existence of edge spectrum.

\subsection{Open problems} An open problem is whether the bulk-edge correspondence generalizes to truncations to domains $\Omega$ satisfying $\Fr(\Omega) = \Fr(\Omega^c) = \infty$. We believe that this condition will need to be strengthened to something more quantitative for the bulk-edge correspondence to hold. There are already results that use the K-theoretic and coarse geometry framework \cite{LT22}; it would be nice to provide analytic proofs. 

It has been shown that the bulk-edge correspondence holds when the gap condition ($\MG \cap \Sigma(H_\pm)$ is empty)  is replaced by a mobility gap condition ($H_\pm$ exhibits dynamical localization within $\MG$); see \cite{EGS05}. At this point we do not know if relaxing Assumption \ref{a1} to a mobility gap gives rise to edge spectrum.

\subsection{Notations} We will use the following notations:
\begin{itemize}
    \item $\bfx = (x_1, x_2)$ denotes an element of $\ZZ^2$.
    \item $|\bfx| = |x_1| + |x_2|$ denotes  the $\ell^1$-norm on $\ZZ^2$.
    \item $B_r(\bfx):= \{\bfy\in \ZZ^2: |\bfy - \bfx| \leq r\}$ is the ball of radius $r\in \RR^+$ centered at $\bfx\in \ZZ^2$.
    \item If $A \subset \RR^2$ and $x \in \RR^2$, $d(x,A)$ denotes the distance from $x$ to $A$.
    \item Given an operator $H: \ell^2(\ZZ^2)\to \ell^2(\ZZ^2)$, we let $H(\bfx,\bfy)= \langle H\delta_\bfx,\delta_\bfy\rangle$  be the kernel of $H$. We let $P_\lambda(H):= \mathbb{1}_{(-\infty, \lambda)}(A)$ be the spectral projection below energy $\lambda$.
    \item In the whole paper, $C_\nu$ denotes a constant that can vary from line to line but depends only on the parameter $\nu$ from \S\ref{sec-1.2}.
\end{itemize}

\subsection{Acknowledgements} This problem was motivated in part by a question from G.M. Graf at a lecture by G.C. Thiang \cite{GT20} during the online workshop ``Mathematics of topological insulators" in 2020 at the American Institute of Mathematics. We are very grateful to the staff at AIM and the organizers of the workshop, D. Freed, G.M. Graf, R. Mazzeo and M.I. Weinstein. 

We gratefully acknowledge support from National Science Foundation DMS 2054589 and the Pacific Institute for the Mathematical Sciences. The contents of this work are solely the responsibility of the authors and do not necessarily represent the official views of PIMS.

\section{The main proposition}

We proved Theorem \ref{thm:1} using Theorem \ref{thm:2} in \S\ref{sec-1.3}. 
Theorem \ref{thm:2} will essentially follow from Proposition \ref{prop:2} below. It essentially asserts that two insulators that coincide on a large enough ball (with radius depending on $\nu$ but not on the center of the ball) must have the same bulk conductance.

\begin{assumption}\label{a4}
    $H$ is a selfadjoint, short-range operator on $\ell^2(\ZZ^2;\CC^d)$, such that for some $\lambda \in \RR$ and $\delta \in (0,1)$,
\begin{equation}
    (\lambda-\delta,\lambda+\delta) \cap \Sigma(H) = \emptyset.
\end{equation}
\end{assumption}

 \begin{prop}\label{prop:2} There exists a constant $C_\nu \geq 1$ such that the following holds. Let $\epsilon > 0$, $r > 0$, $\bfn \in \ZZ^2$, and $H_1$, $H_2$ satisfy Assumption \ref{a4} such that 
    \begin{equation}
        \label{eq: H1_H2_coincide}
    |(H_1 - H_2)(\bfx,\bfy)| \leq \epsilon, \quad \bfx, \bfy \in B_{4r}(\bfn).
    \end{equation}
    Then we have: 
    \begin{equation}
        \label{eq: sigma_B_locally_determined}
    |\sigma(H_1,\lambda) - \sigma(H_2,\lambda)|\leq \frac{C_\nu}{\delta^{12}} \left(e^{-\frac{\delta r}{2C_\nu}} + \epsilon^{1/2}\right).
    \end{equation}
 \end{prop}

\begin{proof}[Proof of Theorem \ref{thm:2} assuming Proposition \ref{prop:2}] 
1. We recall that $\nu$ is a fixed parameter. In this first step, we set the values the constants $\alpha_\nu$ and $R_\nu$. Let $C_\nu$ be given by Proposition \ref{prop:2}; we set $\alpha_\nu = 300 C_\nu$. We now define $R_\nu$. The quantity $R^{12} e^{-\nu R/2}$ goesto $0$ as $R$ goes to infinity. Therefore, there exists $R_\nu \geq 4$ such that for all $R \geq R_\nu$,
\begin{equation}\label{eq-2e}
    R^{12} e^{-\nu R/2} < \dfrac{1}{2}.
\end{equation}

Fix now $R \geq R_\nu$ (in particular, \eqref{eq-2e} holds); and define 
\begin{equation}\label{eq-2c}
\delta = \alpha_\nu \frac{\ln R}{3 R} = 100 C_\nu \dfrac{\ln R}{R}.
\end{equation}
We will prove that $\MG \cap \Sigma(H_e)$ is $3\delta$-dense within $\MG$, that is, $\alpha_\nu \frac{\ln R}{R}$-dense within $\Sigma(H_e)$. 

2. Let us assume for now that the following statements hold,
\begin{align}
    (\lambda_*-\delta, \lambda_* + \delta) \cap \Sigma(H_e) = \emptyset,\\
   \label{eq-2d} (\lambda_*-\delta, \lambda_* + \delta) \cap \Sigma(H_+) = \emptyset,\\
    (\lambda_*-\delta, \lambda_* + \delta) \cap \Sigma(H_-) = \emptyset,
\end{align}
and let's aim for a contradiction. Note that these statements imply that $H_e, H_\pm$ satisfy Assumption \ref{a4}. 

Since $\Omega$ has filling radius at least $R$, there exists $\bfn \in \ZZ^2$ such that $B_{8r}(\bfn) \subset \Omega$, $r = R/8$. We now look at $(H_e-H_+)(\bfx,\bfy)$ for $\bfx,\bfy$ in $B_{4r}(\bfn)$. Because $\bfx, \bfy \in \Omega$, we have:
\begin{equation}
    (H_e-H_+)(\bfx,\bfy) = (H_e-\1_\Omega H_+ \1_\Omega - \1_{\Omega^c} H_- \1_{\Omega^c})(\bfx,\bfy) = E(\bfx,\bfy),
\end{equation}
where $E$ is the operator defined in Assumption \ref{a2}. Moreover, 
\begin{equation}
    d(\bfx,\partial \Omega) \geq d(\bfn,\partial \Omega) - |\bfx-\bfn| \geq 8r- 4r = 4r,
\end{equation}
because $B_{8r}(\bfn) = B_R(\bfn) \subset \Omega$. It follows from \eqref{eq-1a} that 
\begin{equation}
    \big| (H_e-H_+)(\bfx,\bfy) \big| \leq \nu^{-1} e^{- 8\nu r}, \qquad \bfx, \bfy\in B_{4r}(\bfn).
\end{equation}

Proposition \ref{prop:2} then yields
\begin{equation}
    \big| \sigma(H_e,\lambda_*) - \sigma(H_+,\lambda_*) \big| \leq \dfrac{C_\nu}{\delta^{12}} \exp\left(-\dfrac{\delta r}{2C_\nu}\right) + \dfrac{C_\nu \nu^{-1/2}}{\delta^{12}}e^{-4 \nu r}.
\end{equation}
We recall that $\delta$ has the value \eqref{eq-2c}. Therefore, since $C_\nu \geq 1$ and $R \geq R_\nu \geq 4$,
\begin{align}
    \dfrac{C_\nu}{\delta^{12}} \exp\left(-\dfrac{\delta r}{2C_\nu}\right) & = \dfrac{C_\nu R^{12}}{(100 C_\nu \ln R)^{12}} \exp\left(-\dfrac{100 C_\nu \ln R}{R} \cdot \dfrac{R}{16C_\nu}\right) 
    \\
    & \leq R^{12} \exp\left(-\dfrac{25}{4} R\right) \leq R^{-1/2} \leq \dfrac{1}{2}.
\end{align}
The last two inequalities are true for $R\geq 4$. Likewise, because $R$ satisfies \eqref{eq-2e},
\begin{equation}
   \dfrac{C_\nu \nu^{-1/2}}{\delta^{12}}e^{-4\nu r} = \dfrac{C_\nu R^{12}}{(100 C_\nu \ln R)^{12}} \exp\left( -\dfrac{\nu R}{2} \right) \leq R^{12} \exp\left( -\dfrac{\nu R}{2} \right) < \dfrac{1}{2}.
\end{equation}
Going back to \eqref{eq-2d}, we conclude that
\begin{equation}
    \big| \sigma(H_e,\lambda_*) - \sigma(H_+,\lambda_*) \big| < 1. 
\end{equation}
Since bulk conductances are integers (see \cite{EGS05}*{Proposition 3} and Remark \ref{rmk: interger_conductance} below), we conclude that $\sigma(H_e,\lambda_*) = \sigma(H_+,\lambda_*)$.

Similarly, we conclude that $\sigma(H_e,\lambda_*) = \sigma(H_-,\lambda_*)$. This cannot be true, since $\sigma(H_+,\lambda_*) \neq \sigma(H_-,\lambda_*)$. We conclude that for each $\lambda_* \in \MG$, one of the statements among \eqref{eq-2d} must fail. In other words, for all $\lambda_* \in \MG$, there exists some $\lambda \in \Sigma(H_e) \cup \Sigma(H_+) \cup \Sigma(H_-)$ such that $|\lambda-\lambda_*| \leq \delta$.

It remains to show that $\Sigma(H_e) \cap \MG$ is $3\delta$-dense within $\MG$. Write $\MG = (a,b)$ with $b - a>3\delta$ (otherwise any subset is $3\delta$-dense by definition). Let $\lambda_* \in (a+\delta,b-\delta)$ and $\lambda \in \Sigma(H_e) \cup \Sigma(H_+) \cup \Sigma(H_-)$ such that $|\lambda-\lambda_*| \leq \delta$. In particular, $\lambda \in (a,b) = \MG$, which is a spectral gap of $H_\pm$, so $\lambda \in \Sigma(H_e)$. Let now $\lambda_* \in (a,\mu_*)$, $\mu_* = a+2\delta$; since $b-a > 3\delta$, $\mu_* \in (a+\delta,b-\delta)$ and by the previous step there exists $\mu \in \Sigma(H_e)$ such that $|\mu-\mu_*| \leq \delta$. In particular,  $|\mu-\lambda_*| \leq 3\delta$. A similar argument works for $\lambda_* \in (b-2\delta,b)$. We conclude that $\MG \cap \Sigma(H_e)$ is $3\delta$-dense within $\MG$.
\end{proof}

\section{Proof of Proposition \ref{prop:2}}

\subsection{On short-range Hamiltonians}
Throughout the proofs below, we will use the following estimates, proved in Appendix \ref{App}: For $a \in (0,1]$, $R>0$, we have
\begin{align}\label{eq-1e}
    \sum_{s \in \ZZ} e^{-2 a|s|} & \leq \dfrac{2}{a}, \quad \sum\limits_{\bfx\in \ZZ^2} e^{-2a|\bfx|} \leq \dfrac{4}{a^2}, \quad \text{and} \quad \sum_{|\bfx| \geq R, \bfx\in\ZZ^2}  e^{-2a|\bfx|} \leq  \dfrac{8}{a^2} e^{-a R}. 
\end{align}

We make here a few observations on the selfadjoint, short-range Hamiltonians $H$ on $\ell^2(\ZZ^2,\CC^d)$. First, they are bounded in terms of the (fixed) parameter $\nu \in (0,1]$ quantifying the short-range condition \eqref{eq-1c}. Specifically, an application of Schur's test gives:
\begin{equation}\label{eq-1g}
    \| H \| \leq \dfrac{4}{\nu^3}.
\end{equation}
We refer to Appendix \ref{App} for the proof.

As in \cites{EGS05,AW15}, we introduce
\begin{equation}
    S_\alpha:= \sup_{\bfx\in \ZZ^2} \sum_{\bfy\in \ZZ^2} |H(\bfx, \bfy)|\left(e^{\alpha |\bfx - \bfy|} - 1\right).
\end{equation}
We note that if $H$ is short range under the definition \eqref{eq-1c}, then for any $\alpha\in(0,2\nu)$, $S_\alpha<+\infty$. Also, for later use, if 
 $\alpha \in (0,\nu]$:
\begin{equation}\label{eq-1h}
    \dfrac{S_\alpha}{\alpha} \leq \dfrac{S_\nu}{\nu} \leq \dfrac{16}{\nu^4},
\end{equation}
Again, see Appendix \ref{App} for the proof.

We recall the Combes--Thomas inequality \cite{CT73}:

\begin{prop}\label{prop: CTthm}\cite{AW15}*{Theorem 10.5}
    Let $H$ be a selfadjoint, short-range operator on $\ell^2(\ZZ^2;\CC^d)$. If $\alpha \in (0,2\nu)$ and $z \in \CC$ are such that $\Delta := d(z,\Sigma(H)) > S_\alpha$, then we have 
    \begin{equation}
        \label{eq: CTthm_1}
    |(H-z)^{-1}(\bfx, \bfy)|\leq \frac{1}{\Delta- S_\alpha}\exp(-\alpha|\bfx - \bfy|).   
    \end{equation}
\end{prop}

\subsection{Spectral projections}

An application of the Combes--Thomas inequality is control of the kernels of spectral projections:

\begin{lemma}\label{lem:1a} There exists a constant $C_\nu$, such that for any $H, \lambda$ satisfying Assumption \ref{a4}:
\begin{equation}\label{eq-1i}
     \big| P_\lambda(H)(\bfx,\bfy) \big| \leq \frac{C_\nu}{\delta}\exp\left(-\dfrac{\delta}{C_\nu}|\bfx - \bfy|\right).
\end{equation}
\end{lemma}

\begin{proof} 
Set $\alpha = 2^{-5} \nu^4 \delta$,  so that $\alpha \leq \nu$ 
and $S_\alpha \leq \frac{16\alpha}{\nu^4}\leq \frac{\delta}{2}$, see \eqref{eq-1h}. Let $\gamma$ be a contour enclosing $\Sigma(H) \cap (-\infty,\lambda)$, at least $\delta$-distant from $\Sigma(H)$. For $z \in \gamma$, we have:
\begin{equation}\label{eq-1j}
    (H-z)^{-1}(\bfx,\bfy) \leq \dfrac{e^{-\alpha |\bfx-\bfy|}}{\Delta-S_\alpha} \leq \frac{e^{-\alpha |\bfx-\bfy|}}{\delta - \frac{\delta}{2}}\leq \frac{2}{\delta}\exp\left(-\alpha|\bfx - \bfy|\right). 
\end{equation}
Integrating this over $\gamma$, we have:
\begin{equation}
   \big| P_\lambda(H)(\bfx,\bfy) \big| = \left| \dfrac{1}{2\pi i} \oint_\gamma (H-z)^{-1}(\bfx,\bfy) dz \right| \leq \frac{|\gamma|}{\pi\delta}\exp\left(-\alpha|\bfx - \bfy|\right). 
\end{equation}

Without loss of generalities, we may assume that $|\gamma| \leq 2\pi + 4 \| H \|$ -- see Figure \ref{fig:3} 
From the bound \eqref{eq-1g} on $\| H \|$, we deduce that
\begin{equation}\label{eq-1k}
    \dfrac{|\gamma|}{\pi} \leq 2 + \dfrac{16}{\pi\nu^3} \leq \dfrac{8}{\nu^3}.
\end{equation}
We conclude that 
\begin{equation}
    \big| P_\lambda(H)(\bfx,\bfy) \big| \leq \frac{8}{\nu^3\delta}\exp\left(-\dfrac{\nu^4 \delta}{32}|\bfx - \bfy|\right).
\end{equation}
 This yields \eqref{eq-1i} (with for instance $C_\nu =32 \nu^{-4}$).
\end{proof}

\begin{figure}[t]
  \centering
  \includegraphics[width=1\textwidth]{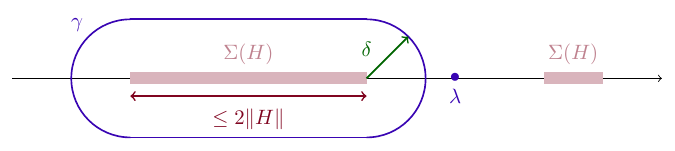}
  \caption{The contour $\gamma$. We note that the spectrum of $H$ is contained in $[-\|H\|,\|H\|]$, so it has diameter less than $2\|H\|$. Since $\delta \leq 1$, the two half-circles have perimeter less than $2\pi$ and the length of $\gamma$ is less than $4\|H\|+2\pi$. Note that $\gamma$ does not need to enclose $\lambda$.}
  \label{fig:3}
\end{figure}

\begin{remark}\label{rmk: interger_conductance}
    As a result, if $H$, $\lambda$ satisfy Assumption \ref{a4}, then the open bounded interval $(\lambda - \delta, \lambda + \delta)$ satisfies condition \cite{EGS05}*{(1.2)}; and it follows that $\sigma(H, \lambda)$ is well-defined and is an integer \cite{EGS05}*{Proposition 3}. 
\end{remark}

\begin{lemma}\label{lem:1b} There exists a constant $C_\nu$ such that the following holds. Let $r, \epsilon > 0$, and two triplets $H_1, \lambda, \delta$ and $H_2, \lambda, \delta$ satisfying Assumption \ref{a4}, such that for $\bfx,\bfy \in B_{4r}(0)$,
\begin{equation}\label{eq-1l}
    |H_1(\bfx,\bfy)-H_2(\bfx,\bfy)| \leq \epsilon.
\end{equation}
Then for $(\bfx,\bfy)$ in $B_{2r}(0)$:
\begin{equation}\label{eq-1n}
    |P_\lambda(H_1)(\bfx,\bfy)-P_\lambda(H_2)(\bfx,\bfy)| \leq \dfrac{C_\nu}{\delta^6} \left( \exp\left(-\dfrac{\delta r}{C_\nu}\right) + \epsilon \right).
\end{equation}
\end{lemma}

\begin{proof} Let $\alpha, \gamma$ be as in the proof of Lemma \ref{lem:1a}. We have:
\begin{equation}
   \big(P_\lambda(H_1) - P_\lambda(H_2)\big) (\bfx,\bfy) = \sum_{\bfx',\bfy'} \dfrac{1}{2\pi i} \oint_\gamma (H_2-z)^{-1}(\bfx,\bfx') \big( H_2(\bfx',\bfy') - H_1(\bfx',\bfy') \big) (H_1-z)^{-1}(\bfy',\bfy) dz. 
\end{equation}
From \eqref{eq-1j} then \eqref{eq-1k}, and using the bound \eqref{eq-1k} for $|\gamma|$:
\begin{align}
     \big| \big(P_\lambda(H_1) - P_\lambda(H_2)\big) (\bfx,\bfy) \big| & \leq \dfrac{2|\gamma|}{\pi \delta^2}\sum_{\bfx',\bfy'}  \exp\left(-\alpha|\bfx - \bfx'|-\alpha|\bfy - \bfy'|\right) \big| H_2(\bfx',\bfy') - H_1(\bfx',\bfy') \big|
    \\
    & \leq \dfrac{2^4}{\nu^3\delta^2 }\sum_{\bfx',\bfy'}  \exp\left(-\alpha|\bfx - \bfx'|-\alpha|\bfy - \bfy'|\right) \big| H_2(\bfx',\bfy') - H_1(\bfx',\bfy') \big|. \label{eq-1f}
\end{align}
We split the RHS sum in two parts, depending whether $\bfx', \bfy' \in B_{4r}(0)$ or not. When they do, we can use the bound \eqref{eq-1l} on the kernel of $H_1-H_2$. It yields
\begin{align}
    & \sum_{x',y' \in B_{4r}(0)}  \exp\left(-\alpha|\bfx - \bfx'|-\alpha|\bfy - \bfy'|\right) \big| H_2(x',y') - H_1(x',y') \big|
    \\
   \leq &  \epsilon \sum_{x',y' \in \ZZ^2}  \exp\left(-\alpha|\bfx - \bfx'|-\alpha|\bfy - \bfy'|\right) 
    \\
   \leq & \epsilon \sum_{x',y' \in \ZZ^2}  \exp\left(-\alpha|\bfx'|-\alpha|\bfy'|\right) \leq  \epsilon \left( \dfrac{16}{\alpha^2} \right)^2, 
\end{align}
where in the last line we used \eqref{eq-1e}.

When now restricting to $\bfx'$ or $\bfy' \in B_{4r}(0)^c$, we note that for  $\bfx, \bfy \in B_{2r}(0)$ either $|\bfx - \bfx'| \geq 2r$ or $|\bfy - \bfy'| \geq 2r$. Thus, we have
\begin{align}
    & \sum_{\bfx' \text{ or } \bfy' \in B_{2r}(0)^c}  \exp\left(-\alpha|\bfx - \bfx'|-\alpha|\bfy - \bfy'|\right) \big|H_2(\bfx',\bfy') - H_1(\bfx',\bfy') \big|
    \\
    \leq &  \dfrac{2}{\nu} \sum_{\bfx'  \text{ or } \bfy' \in B_{4r}(0)^c}  \exp\left(-\alpha|\bfx - \bfx'|-\alpha|\bfy - \bfy'|\right) \leq \dfrac{4}{\nu} \sum_{\bfx' \in B_{2r}(0)^c, \bfy' \in \ZZ^2}  \exp\left(-\alpha|\bfx'|-\alpha|\bfy'|\right)
    \\
    & \leq \dfrac{4}{\nu} \left(\frac{32}{\alpha^2}e^{-\alpha r}\right) \left(\frac{16}{\alpha^2}\right) = \frac{8}{\nu} \left(\frac{2^8}{\alpha^4}\right)e^{-\alpha r}.
\end{align}
In the last line, we applied \eqref{eq-1e}. So, heading back to \eqref{eq-1f} and using the value $\alpha = 2^{-5} \nu^4 \delta$ from the proof of Lemma \ref{lem:1a}, we obtain
\begin{equation}
    \big| \big(P_\lambda(H_1) - P_\lambda(H_2)\big) (x,y) \big| \leq \dfrac{2^7}{\nu^4 \delta^2} \left(\dfrac{2^8}{\alpha^4}\right) (e^{-\alpha r}+\epsilon) = \dfrac{2^{35}}{\nu^{20} \delta^6} \left(\exp\left(-\dfrac{\nu^4 \delta}{32}r\right) + \epsilon\right).
\end{equation}
This yields \eqref{eq-1n} (with $C_\nu = 2^{35}\nu^{-20}$ -- we made no attempts to minimize this constant).     
\end{proof}

\subsection{Technical result.} The key technical step in the proof of Proposition \ref{prop:2} is:

\begin{prop}\label{prop:3} Fix $\varepsilon > 0, \ C > 0, \ r >0, \ \beta \in (0,1]$. Let $A_0, A_1, A_2$ be three operators on $\ell^2(\ZZ^2,\CC^d)$ with the following properties:
\begin{itemize}
    \item[(i)] For $j \in \{0,1,2\}$, $|A_j(\bfx,\bfy)| \leq C e^{-2\beta|x-y|}$;
    \item[(ii)] There exists $k \in \{0,1,2\}$ such that if $\bfx,\bfy \in B_{2r}(0)$,then $\big| A_k(\bfx,\bfy) \big| \leq C \varepsilon$.
\end{itemize}
Then $B := A_0[A_1,\Lambda_1][A_2,\Lambda_2]$ is trace-class and 
\begin{equation}
   \big| \Tr(B)  \big| \leq  C^3 \frac{2^{16}}{\beta^6} \left(e^{-\beta r} + \varepsilon^{1/2} \right).
\end{equation}
\end{prop}

Let us start with a simple result:

\begin{lemma}\label{lemma: all} Let $A$ be a bounded operator on $\ell^2(\ZZ^2,\CC^d)$ with $|A(\bfx,\bfy)| \leq e^{-2\beta |\bfx-\bfy|}$. Then 
    \begin{equation}\label{eq-1x}
        \big| [A, \Lambda_1](\bfx, \bfy) \big|\leq e^{-2\beta|x_1| - 2\beta|y_1| - 2\beta|x_2 - y_2|}.
    \end{equation}
If moreover $|A(\bfx,\bfy)| \leq \varepsilon$ for $(\bfx,\bfy) \in B_{2r}(0)$, then  
     \begin{equation}\label{eq-1y}
        \big| [A, \Lambda_1](\bfx, \bfy) \big| \leq \varepsilon^{1/2} e^{-\beta|x_1| - \beta|y_1| - \beta|x_2 - y_2|}, \qquad \bfx,\bfy \in B_{2r}(0).
    \end{equation}
\end{lemma}

\begin{proof} The kernel of $[A, \Lambda_1]$ is 
\begin{equation}
    [A, \Lambda_1](\bfx, \bfy) = A(\bfx,\bfy)(\Lambda_1(\bfy) - \Lambda_1(\bfx)).
\end{equation}
We note that $|\Lambda_1(\bfy) - \Lambda_1(\bfx)| = 0$ if $x_1, y_1$ are both positive or both negative; and it is at most $1$ otherwise, that is if $x_1y_1 \leq 0$. Therefore, we have the bound
\begin{equation}
    \big| [A, \Lambda_1](\bfx, \bfy)\big| \leq C e^{-2\beta |\bfx-\bfy|} \mathbb{1}_{x_1 y_1 \leq 0}.
\end{equation}
Whenever $x_1y_1 \leq 0$, we have
\begin{equation}
    |\bfx-\bfy| = |x_1-y_1| + |x_2-y_2| = |x_1| + |y_1| + |x_2-y_2|.
\end{equation}
It follows that 
\begin{equation}
    \big| [A, \Lambda_1](\bfx, \bfy)  \big| \leq e^{-2\beta |x_2-y_2|-2\beta|x_1|-2\beta|y_1|}.
\end{equation}
This completes the proof of \eqref{eq-1x}. To prove \eqref{eq-1y}, we recall that $|\Lambda_1(\bfy) - \Lambda_1(\bfx)| \leq 1$; hence $\big| [A, \Lambda_1](\bfx, \bfy) \big| \leq \varepsilon$. It suffices then to interpolate this bound with \eqref{eq-1x}.
\end{proof}

For the proof of Proposition \ref{prop:3}, we will use the following inequality:
for $\beta \in (0,1]$, $\bfx, \bfw\in \ZZ^2$, 
\begin{equation}\label{eq-1r}
    \sum_{\bfy,\bfz \in \ZZ^2} e^{-2\beta |\bfx-\bfy| - 2\beta |y_2 - z_2| -2\beta |y_1| - 2\beta |z_1| - 2\beta |z_1 - w_1| - 2\beta |z_2| - 2\beta |w_2|} \leq \left( \dfrac{4}{\beta} \right)^4 e^{-\beta|\bfx|-\beta|\bfw|}.
\end{equation}
We refer to Appendix \ref{App} for a proof.

\begin{proof}[Proof of Proposition \ref{prop:3}] 1. By a scaling argument, we can assume that $C=1$. We first control the kernel of $B$: 
\begin{equation}
    |B(\bfx,\bfw)| = \left| \sum_{y,z \in \ZZ^2} A_0(\bfx,\bfy)B_1(\bfy,\bfz)B_2(\bfz,\bfx) \right|,
\end{equation}
where $B_j = [A_j,\Lambda_j]$. We control the kernels of $A_0$, $B_1, B_2$ using assumption $(i)$ and \eqref{eq-1x}. It yields:
\begin{align}
    |B(\bfx,\bfw)| & = \left| \sum_{\bfy,\bfz \in \ZZ^2} A_0(\bfx,\bfy)B_1(\bfy,\bfz)B_2(\bfz,\bfw) \right|
    \\
    & \leq \sum_{\bfy,\bfz \in \ZZ^2} e^{-2\beta |\bfx-\bfy| - 2\beta |y_2-z_2| -2\beta |y_1| - 2\beta |z_1| - 2\beta |z_1-w_1| - 2\beta |z_2| - 2\beta |w_2|} 
    \\
    & \leq \left( \dfrac{4}{\beta} \right)^4 e^{-\beta|\bfx|-\beta|\bfw|}.\label{eq-1s}
\end{align}
Thus $|B(\bfx,\bfw)|$ decays exponentially, hence $B$ is trace-class; moreover 
\begin{align} 
    \Tr(B) \leq  \sum_{\bfx \in \ZZ^2} |B(\bfx,\bfx)| \leq & \sum_{\bfx,\bfy,\bfz \in \ZZ^2} e^{- 2\beta |\bfx-\bfy| - 2\beta |y_2 - z_2| - 2\beta |y_1| - 2\beta |z_1| - 2\beta |z_1 - x_1| - 2\beta |z_2| - 2\beta |x_2|}\\
    := & \sum_{\bfx,\bfy,\bfz \in \ZZ^2} f(2\beta, \bfx, \bfy,\bfz). \label{eq-1t}
\end{align}

2. We now split the sum in \eqref{eq-1t} in two pieces: $|\bfx| \geq r$ and $|\bfx| \leq r$. Thanks to \eqref{eq-1r} and \eqref{eq-1e}, we have 
\begin{equation}
    \sum_{|\bfx| \geq r} f(2\beta, \bfx, \bfy,\bfz) \leq \left( \dfrac{4}{\beta} \right)^4 \sum_{|\bfx| \geq r} e^{-2\beta|\bfx|} \leq \left( \dfrac{4}{\beta} \right)^4 \dfrac{8}{\beta^2} e^{-\beta r} = \frac{2^{11}}{\beta^6}e^{-\beta r}.
\end{equation}
We focus below on $|\bfx| \leq r$. 

3. If $k=0$ in (ii), then we split the sum in \eqref{eq-1t} according to $|\bfy| \geq 2r$ and $|\bfy| \leq 2r$. In the former case, $|\bfx-\bfy| \geq r$. Therefore, when $|\bfx| \leq r, \ |\bfy| \geq 2r$, we deduce that
\begin{align}
    \big| A_0(\bfx,\bfy) \big| & \leq e^{-2\beta|\bfx-\bfy|} \leq  e^{-\beta r -\beta|\bfx-\bfy|},
  \\
  \big| A_0(\bfx,\bfy)B_1(\bfy,\bfz)B_2(\bfz,\bfx) \big| & \leq e^{-\beta r} f(\beta,\bfx,\bfy,\bfz). \label{eq-1v}
\end{align}
If now $|\bfy| \leq 2r$ (and $|\bfx| \leq r \leq 2r$), then we can use (ii). Interpolating with (i) gives, for $|\bfx| \leq r, \ |\bfy| \leq 2r$:
\begin{align}
    \big| A_0(\bfx,\bfy) \big| & \leq \varepsilon^{1/2} e^{-\beta|\bfx-\bfy|}, 
    \\
  \big| A_0(\bfx,\bfy)B_1(\bfy,\bfz)B_2(\bfz,\bfx) \big| & \leq \varepsilon^{1/2} f(\beta,\bfx,\bfy,\bfz).\label{eq-1u}
\end{align}
Summing the bounds \eqref{eq-1v} and \eqref{eq-1u} produces:
\begin{align}
    \sum_{\substack{|\bfx| \leq r, \\ \bfy, \bfz \in \ZZ^2}} \big| A_0(\bfx,\bfy)B_1(\bfy,\bfz)B_2(\bfz,\bfx) \big| & \leq e^{-\beta r} \sum_{\substack{ |\bfx| \leq r, \\ |\bfy| \geq 2r, \bfz \in \ZZ^2}} f(\beta,\bfx,\bfy,\bfz) + \varepsilon^{1/2} \sum_{\substack{|\bfx| \leq r, \\ |\bfy| \leq 2r, \bfz \in \ZZ^2}} f(\beta,\bfx,\bfy,\bfz)
    \\
    & \leq \left(e^{-\beta r} + \varepsilon^{1/2} \right) \sum_{\bfx,\bfy,\bfz \in \ZZ^2}  f(\beta,\bfx,\bfy,\bfz) 
    \\
    & \leq  \frac{2^{16}}{\beta^6} \left(e^{-\beta r} + \varepsilon^{1/2} \right) . \label{eq-1w}
\end{align}
where we used \eqref{eq-1r} and \eqref{eq-1e} to get
\begin{equation}
    \sum_{\bfx,\bfy,\bfz \in \ZZ^2} f(\beta,\bfx,\bfy,\bfz) \leq  \left(\frac{8}{\beta}\right)^4 \sum_{\bfx \in \ZZ^2} e^{-\beta|\bfx|} \leq \left(\frac{8}{\beta}\right)^4 \frac{16}{\beta^2} = \frac{2^{16}}{\beta^6}. \label{eq-1z}
\end{equation}

4. We now work on $k=1$. We split the sum in $\bfy, \bfz \in B_{2r}(0)$ and $\bfy$ or $\bfz$ outside $B_{2r}(0)$. In the latter case, either $|\bfx-\bfy| \geq r$ or $|\bfz-\bfx| \geq r$. So either 
\begin{equation}
    \big|A_0(\bfx,\bfy) \big| \leq e^{-\beta r} e^{-\beta|\bfx-\bfy|}  \text{ or } \big| B_2(\bfz,\bfx) \big| \leq e^{-\beta r}  e^{-\beta|\bfz-\bfx|}.
\end{equation}
In either case, we recover the bound \eqref{eq-1v}. In the former case, we can use \eqref{eq-1y}, and recover \eqref{eq-1u}. Since \eqref{eq-1v} and \eqref{eq-1u} lead to \eqref{eq-1w}, we obtain this bound here as well. 

5. The case $k=2$ follows the same path as $k=0$ . This completes the proof.  
\end{proof}

\subsection{Comparison of bulk conductances} We will use the following result \cite{EGS05}*{Lemma 7(ii)}, which essentially states that the bulk conductance is independent of $\Lambda_1, \Lambda_2$:

\begin{prop}\label{prop:1} Let $H$ be a short-range operator on $\ell^2(\ZZ^2,\CC^d)$ and $\lambda \notin \Sigma(H)$. For any $\bfn$,
\begin{equation}
    \sigma(H,\lambda) = -2\pi i\Tr\left( P_\lambda(H) \big[ [P_\lambda(H),\Lambda_1(\cdot-n_1)], [P_\lambda(H),\Lambda_2(\cdot-n_2)] \big]\right).
\end{equation}
\end{prop}

We are now ready to prove Proposition \ref{prop:2}.

\begin{proof}[Proof of Proposition \ref{prop:2}] 1. For simplicity, use the notation $P_j = P_\lambda(H_j)$. Let $T$ be the translation by $\bfn$: $T u (\cdot)= u(\cdot - \bfn)$. We have $\Lambda_j(\cdot - n_j) = T \Lambda_j(\cdot) T^*$ and $T^* P_\lambda(H) T = P_\lambda(T^* HT)$. Using these and Proposition \ref{prop:1}, as well as the cyclicity of the trace, we obtain
\begin{align}
    \sigma(H_j,\lambda) & = -2\pi i\Tr \big( P_j \big[ [P_j, T \Lambda_1 T^*], [P_j, T \Lambda_2 T^*] \big] \big)
    \\
    & = -2\pi i\Tr \big(  P_j  \big[ T [ T^* P_j T,  \Lambda_1 ] T^*, T[T^*  P_j T,   \Lambda_2] T^* \big] \big)
    \\
    & = -2\pi i\Tr \big(  P_j   T \big[ [ T^*P_j T,  \Lambda_1 ] , [T^*  P_j T,   \Lambda_2] \big] T^* \big) 
    \\
    & = -2\pi i\Tr \big(  T^* P_j   T \big[ [ T^* P_j T,  \Lambda_1 ] , [T^*  P_j T,   \Lambda_2] \big]  \big) = \sigma(T^* H_j T,\lambda).  
\end{align}
Therefore, by replacing $H_j$ by $T^* H_j T$, we can simply assume that $\bfn = 0$. 

2. Now we write 
\begin{equation}\label{eq-2a}
    \sigma(H_1,\lambda)-\sigma(H_2,\lambda) = T(P_1-P_2,P_1,P_1) + T(P_2,P_1-P_2,P_1) + T(P_2,P_2,P_1-P_2), 
\end{equation}
where $T(A_0,A_1,A_2)$ is the trilinear form
\begin{equation}
    T(A_0,A_1,A_2) = -2\pi i \Tr\left( A_0 [A_1,\Lambda_1] [A_2,\Lambda_2]\right) + 2\pi i \Tr\left( A_0 [A_2,\Lambda_2] [A_1,\Lambda_1]\right). 
\end{equation}

From Lemmas \ref{lem:1a} and \ref{lem:1b}, we have (for a constant $C_\nu$ depending on $\nu$ only):
\begin{align}
    \big| P_j(\bfx,\bfy) \big| & \leq \dfrac{C_\nu}{\delta} e^{-\delta |\bfx-\bfy|/C_\nu},  \qquad j=1,2;
    \\
    \big| P_1(\bfx,\bfy) - P_2(\bfx,\bfy) \big| & \leq \dfrac{C_\nu}{\delta^6} \left( e^{-\delta |\bfx-\bfy|/C_\nu} + \epsilon\right), \qquad \bfx,\bfy \in B_{2r}(\bfn).
\end{align}
Therefore, the triplets $(A_0,A_1,A_2)$ involved in \eqref{eq-2a} satisfy the assumptions of Proposition \ref{prop:3}, with constants
\begin{equation}
    C = \dfrac{C_\nu}{\delta}, \quad \beta = \dfrac{\delta}{2 C_\nu}, \quad \varepsilon = \dfrac{1}{\delta^5} \left( e^{-\delta r / C_\nu} + \epsilon\right).
\end{equation}
So, we deduce that
\begin{equation}
    \big|\sigma(H_1,\lambda))-\sigma(H_2,\lambda)\big| \leq \dfrac{C_\nu}{\delta^{3+6}} \left( e^{\frac{\delta r}{2C_\nu}} + \dfrac{1}{\delta^{5/2}} (e^{-\frac{\delta r}{C_\nu}} + \epsilon)^{1/2} \right) \leq \dfrac{C_\nu}{\delta^{12}} \left( e^{-\frac{\delta r}{2C_\nu}} + \epsilon^{1/2} \right).
\end{equation}
This completes the proof of Proposition \ref{prop:2}.
\end{proof}

\section{Violation of the bulk-edge correspondence in strips}

In this section, we show that topological insulators lying within strips do not necessarily support edge states along their boundary; this means that geometrically, $\Omega$ needs to be unbounded in all directions for the bulk spectral gaps to systematically fill. 

Specifically, for any $L > 0$, 
we construct an edge operator $H_e$ satisfying Assumption \ref{a2} with:
\begin{itemize}
    \item The bulk operators $H_\pm$ are insulating at energy $0$, with bulk conductance $\mp 1$;
    \item $\Omega \subset \ZZ \times [-L,L]$, in particular $\Fr(\Omega) \leq L$; 
    \item $0 \notin \Sigma(H_e)$: the bulk gap did not fully close.
\end{itemize}
Hence, although the bulk operators $H_\pm$ represent topologically distinct topological phases, the interface $\partial \Omega$ does not support conducting states for $H_e$. In particular, a material made of topologically distinct insulators across $\partial \Omega$, $\Omega = \NN \times [-L,L]$, violates the bulk-edge correspondence. This was suspected by G.M. Graf, but the problem was left open in an online talk by G.C. Thiang \cite{GT20}.

\subsection{Haldane model} Our bulk operators are based on Haldane's model \cite{H88}, which we review briefly. For each site of the honeycomb lattice, the Haldane Hamiltonian models tunneling to the three nearest neighbors and \textit{complex} coupling (for instance induced by a periodic magnetic field) to the six second-nearest neighbors, see Figure \ref{fig:4}. We will use a version based on the $\ZZ^2$-lattice (which only differ from the standard honeycomb version by a linear change of variable):
\[
    H_\pm = H_0 \pm S
\]
where $H_0, S$ are selfadjoint, short-range Hamiltonians on $\ell^2(\ZZ^2,\CC^d)$ given by:
\[
\begin{split}
&(H_0\psi)_n = \begin{bmatrix}
    \psi_n^B + \psi_{n - e_1}^B + \psi_{n - e_2}^B\\
    \psi_n^A + \psi_{n + e_1}^A + \psi_{n + e_2}^A
\end{bmatrix},\\
&(S\psi)_n = is\begin{bmatrix}
    \psi_{n + e_1}^A - \psi_{n - e_1}^A + \psi_{n - e_2}^A - \psi_{n + e_2}^A + \psi_{n + e_2 - e_1}^A - \psi_{n + e_1 - e_2}^A\\
    -\psi_{n + e_1}^B + \psi_{n - e_1}^B - \psi_{n - e_2}^B + \psi_{n + e_2}^B - \psi_{n + e_2 - e_1}^B + \psi_{n + e_1 - e_2}^B
\end{bmatrix}.
\end{split}
\]
The parameter $s$ above quantifies the ratio between first and second-nearest neighbor coupling. We restrict it to $(0,1]$ here.

\begin{figure}
    \centering
    \begin{minipage}{0.4\textwidth}
        \includegraphics[width=0.8\textwidth]{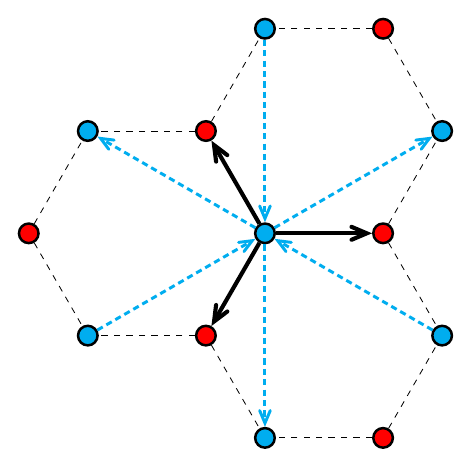}
  \end{minipage}
    \begin{minipage}{0.55\textwidth}
        \centering
        \caption{The black arrows represent tunneling to the three nearest neighbors; the dashed blue arrows represent complex coupling to six second-nearest neighbors in the Haldane model \cite{H88}.}
        \label{fig:side_by_side_p2}
        \label{fig:4}
    \end{minipage}
\end{figure}

As a result, the discrete Fourier transform w.r.t. $\ZZ^2$ is
\[
\begin{split}
\hat{H}_\pm(\xi) & =  \begin{bmatrix}
    \pm 2s\eta(\xi)  & \overline{\omega(\xi)}\\
    \omega(\xi) & \mp  2s \eta(\xi)
\end{bmatrix}, \qquad \xi \in [-\pi, \pi]^2,
\\
\omega(\xi) & := 1+e^{i\xi_1} + e^{i\xi_2}, \\
\eta(\xi) & := \sin(\xi_1) - \sin(\xi_2)  + \sin(\xi_2 - \xi_1).
\end{split}
\]
The eigenvalues of $\hat{H}_\pm(\xi)$ are $\pm \lambda(\xi)$ with
\begin{equation}\label{eq-1b}
    \lambda(\xi) = \sqrt{(2s\eta(\xi))^2 + |\omega(\xi)|^2}.
\end{equation}
The functions $\eta$ and $\omega$ do not vanish simultaneously; therefore $\lambda$ never vanishes and the Hamiltonians $H_\pm$ are insulating at energy $0$:  $0 \notin \Sigma(H_\pm)$, equivalently $H_\pm$ are invertible. Because of translational invariance, it is known their bulk conductance equals the Chern number of their low energy eigenbundle (see e.g. \cite{S15}*{Equations (13) and Corollary 8.4.4}) and we have $\sigma(H_\pm; 0) = \mp 1$, see \cite{BH13}*{\S8}. In particular, the insulators described by $H_\pm$ are topologically distinct and $H_\pm$ satisfy Assumption \ref{a1}.

\subsection{Edge operator} 

Fix $L > 0$ and $\Omega \subset [-L, L]\times \NN$, we define the edge operator by 
\[
H_e := H_+ - 2\1_{\Omega^c}S\1_{\Omega^c}. 
\] 
Then formally, $H_e$ is $H_+$ in $\Omega$ and $H_-$ in $\Omega^c$. We prove Theorem \ref{thm:3}, formulated here as:

\begin{prop}\label{prop: He_invertible} $H_e$ satisfies Assumption \ref{a2} at energy $0$; however there exists a numerical constant $\rho_0 > 0$ such that if $0< s \leq \rho_0 L^{-1}$ then $0 \notin \Sigma(H_e)$. 
\end{prop}

This implies that an interface lying in a strip, between two topologically distinct insulating phases does not necessarily fill the bulk spectral gap. 

\begin{remark} A general argument implies that the edge conductance of $H_e$ across the $x_2$-axis is $0$. Indeed, this conductance is stable under perturbations within strips orthogonal to the $x_2$-axis, such as $2\1_{\Omega^c}S\1_{\Omega^c}$; so it is equal to that of $H_e + 2\1_{\Omega^c}S\1_{\Omega^c} = H_+$, which is $0$. 

To the best of our knowledge, there is no general argument that implies that $H_e$ has edge conductance across the $x_1$-axis equal to $0$. For $0 < s \leq \rho_0 L^{-1}$, it is a consequence of Proposition \ref{prop: He_invertible}: $H_e$ has no state with energy near $0$. For $\Omega = \NN \times [-L,L]$, this implies that no quantum particle may travel from one end of $\partial \Omega$ to the other with high probability. 
\end{remark}

Proposition \ref{prop: He_invertible} is a consequence of the uncertainty principle: a function localized in frequency may not be localized in position. For $s$ small, the Fourier transforms of $H_\pm$ have eigenvalues of order $1$, unless $\xi$ is near the zeroes $\xi_\pm^* = \pm 2\pi/3 (1,-1)$ of $\omega(\xi)$ -- in which case they are of order $s$. Therefore, a $O(s)$-perturbation (such as $S$) may not close the gap unless it generates states for $H_e$ that are concentrated in frequency near $\xi_\pm^*$. By (a taylored version of) the uncertainty principle, such states may not be localized within a strip (such as $\Omega$).

\subsection{Proof} We prove Proposition \ref{prop: He_invertible} here. We will need the following lemmas:

\begin{lemma}\label{lemma: lower_bound_of_lambda}  
    \begin{enumerate}
        \item  There exists $\lambda_0>0$, such that 
        \[
        \lambda(\xi)\geq \lambda_0 \cdot s,  \quad \text{~for~all~} \quad \xi\in [-\pi, \pi]^2, \quad s\in (0,1].
        \]
        \item There exists $\mu_0>0$ such that
        \[
        |\omega(\xi)|\geq \mu_0 \cdot d(\xi, \{ \xi_\pm^*\}), \quad \text{~for~all~} \quad \xi\in [-\pi, \pi]^2.
        \]
    \end{enumerate}
\end{lemma}


\begin{remark} At the physical level these are well-known bounds; we taylor them here to our needs. Part (1) means that $H_\pm$ have a spectral gap at energy $0$; part (2) means that the Wallace Hamiltonian $H_0$ have a Dirac cone.
\end{remark}

\begin{proof}[Proof of Lemma \ref{lemma: lower_bound_of_lambda}]
(1) From \eqref{eq-1b} and $s \in (0,1]$, we have  
\begin{equation}\label{eq-1d}
    \lambda(\xi) = \sqrt{(2s\eta(\xi))^2 + |\omega(\xi)|^2} \geq s \sqrt{(2\eta(\xi))^2 + |\omega(\xi)|^2}.
\end{equation}

Moreover,
    \[
    \begin{split}
    \omega(\xi) = 0 \quad 
    &\Leftrightarrow \quad  \xi = \xi^*_\pm = \pm\dfrac{2\pi}{3}(1,-1) \quad \Rightarrow \quad \eta(\xi^*_\pm) = \pm \frac{3\sqrt{3}}{2}.
    \end{split}
   \]
   Thus $\omega(\xi)$ and $\eta(\xi)$ cannot vanish simultaneously and $\sqrt{(2\eta(\xi))^2 + |\omega(\xi)|^2}$ never vanishes. By continuity, $\sqrt{(2\eta(\xi))^2 + |\omega(\xi)|^2} \geq \lambda_0$ for some $\lambda_0>0$. This proves (1) by going back to \eqref{eq-1d}.

(2) We first write down $\omega$ as a function valued in $\RR^2$ instead of $\CC$:
  \begin{equation}
      \omega(\xi) = (1 + \cos(\xi_1) + \cos(\xi_2), \sin(\xi_1) + \sin(\xi_2)).
  \end{equation}
 With this notation,
 \[
   \nabla(\omega(\xi)) = \begin{bmatrix}
       -\sin(\xi_1) & -\sin(\xi_2) \\ 
       \cos(\xi_1) & \cos(\xi_2)
       \end{bmatrix}
   , \qquad \nabla(\omega(\xi^*_\pm)) = \dfrac{1}{2}\begin{bmatrix}
       \mp\sqrt{3} & \pm\sqrt{3} \\ 
       -1 & -1
   \end{bmatrix}.
   \]
   As a result, for any $u = (u_1, u_2)^T$, 
   \[
   \begin{split}
       | \nabla\omega(\xi^*_\pm) u|^2 &= \frac{1}{4}[3(u_1 - u_2)^2 + (u_1 + u_2)^2]  = \frac{1}{4}(4u_1^2 + 4u_2^2 - 4u_1u_2) \geq \frac{1}{2}| u|^2.
   \end{split}
   \]
   Assume for any $n$, there is $\xi_n\neq \xi_\pm^*\in [-\pi, \pi]^2$ such that 
   \begin{equation}\label{eq-2f}
       |\omega(\xi_n)|\leq \frac{d(\xi_n, \{\xi_\pm^*\})}{n}. 
   \end{equation}
      By compactness of $[-\pi, \pi]^2$, there exists a subsequence of $\xi_n$ that converges to some $\xi_\infty$. From \eqref{eq-2f} and $d(\xi_n, \{\xi_\pm^*\}) \leq 4\pi$, we deduce $|\omega(\xi_\infty)| = 0$ hence $\xi_\infty$ is either $\xi^*_+$ or $\xi_-^*$. As a result, as $n \rightarrow \infty$,
    \[
    \begin{split}
        \frac{1}{n}\geq \frac{|\omega(\xi_n) - \omega(\xi_\infty)|}{|\xi_n - \xi_\infty|}  &= \frac{|\nabla\omega(\xi_\infty)(\xi_n - \xi_\infty)| + O(|\xi_n - \xi_\infty|^2)}{|\xi_n - \xi_\infty|}\\
        &\geq \frac{1}{\sqrt{2}} + O(|\xi_n - \xi_\infty|)\to \frac{1}{\sqrt{2}}.
    \end{split}
    \]
    We get a contradiction. Thus there is some $\mu_0>0$ such that $|\omega(\xi)|\geq \mu_0|\xi - \xi_\pm^*|$.\end{proof}


    \begin{lemma}\label{lemma: H_+_inv_lower_bound} There exists $C_0 > 0$ such that for all $s \in (0,1]$, $L >0$ and $u \in \ell^2(\ZZ^2,\CC^2)$:
        \[
        \supp u \subset \ZZ\times [-L, L] \qquad \Rightarrow \qquad
        \Vert H_+^{-1} u \Vert_{\ell^2} 
         \leq C_0L^{1/3} s^{-2/3}\Vert u \Vert_{\ell^2}.
\]
  \end{lemma}

    \begin{proof}[Proof of Lemma \ref{lemma: H_+_inv_lower_bound}] 
    Recall that $\omega^{-1}(0) = \{\xi^*_+, \xi^*_-\}$. Since eigenvalues of $\hat{H}_+^{-1}(\xi)$ are $\pm\frac{1}{\lambda(\xi)}$, for any $\delta>0$, we have
    \begin{equation}
        \label{eq: H+inv}
        \begin{split}
            \Vert H_+^{-1} u\Vert^2_{\ell^2(\ZZ^2)} &= \Vert \hat{H}_+^{-1}(\xi) \hat{u}(\xi)\Vert_{L^2([-\pi, \pi]^2)}^2 \\
            &\leq \int_{[-\pi,\pi]^2} \frac{1}{(\lambda(\xi))^2} |\hat{u}(\xi)|^2 d\xi\\
            &\leq \int_{B_\delta(\xi^*_\pm)} \frac{1}{\lambda_0^2 s^2}|\hat{u}(\xi)|^2 d\xi+ \int_{\left(B_\delta(\xi^*_\pm)\right)^c}\frac{1}{\mu_0^2 \delta^2}|\hat{u}(\xi)|^2 d\xi\\
            &\leq \int_{B_\delta(\xi^*_\pm)} \frac{1}{\lambda_0^2 s^2}|\hat{u}(\xi)|^2 d\xi + \frac{(2\pi)^2}{\mu_0^2\delta^2}\Vert u \Vert_{\ell^2}^2
        \end{split}
      \end{equation}
      thanks to Plancherel's formula $\Vert \hat{u}\Vert_{L^2([-\pi, \pi]^2)} = 2\pi\Vert u\Vert_{\ell^2(\ZZ^2)}$. If $\supp(u) \subset \ZZ \times [-L, L]$, then 
    \[
    \hat{u}(\xi_1, n_2) = \sum\limits_{n_1} e^{-n_1\xi_1} u(n_1,n_2) = 0, \text{~if~} n_2\notin [-L,L].\]
    Thus we have 
    \[
    \begin{split}
        \int_{B_\delta(\xi^*_\pm)} \frac{1}{\lambda_0^2s^2}|\hat{u}(\xi)|^2 d\xi &\leq \frac{1}{\lambda_0^2}\int_{B_\delta(\xi^*_\pm)} \left| \sum\limits_{n_2\in[-L,L]} e^{-i n_2\xi_2} \hat{u}(\xi_1, n_2) \right|^2 d\xi\\
        &\leq \frac{1}{\lambda_0^2s^2} \int_{B_\delta(\xi^*_\pm)} 2L\sum\limits_{n_2\in[-L,L]} |\hat{u}(\xi_1, n_2)|^2  d\xi \\
        &\leq \frac{2L}{\lambda_0^2s^2}\int_{|\xi_2 - (\xi^*_\pm)_2|\leq \delta}\int_{-\pi}^\pi \sum\limits_{n_2}|\hat{u}(\xi_1, n_2)|^2 d\xi_1 d\xi_2\\
        &=\frac{2L\cdot 4\delta}{\lambda_0^2s^2}  \int_{-\pi}^\pi \sum\limits_{n_2} |\hat{u}(\xi_1, n_2)|^2 d\xi_1 = \frac{8\delta L \cdot 2\pi}{\lambda_0^2s^2} \Vert u \Vert_{\ell^2}
  \end{split}
    \]
   where we use the Plancherel's formula on $n_1$-coordinates only for the last line. Combining with the earlier estimates, we get 
    \[
    \Vert H_+^{-1} u \Vert_{\ell^2(\ZZ^2)}^2 \leq \left(\frac{16\pi\delta L}{\lambda_0^2s^2} + \frac{4\pi^2}{\mu_0^2\delta^2}\right)\Vert u \Vert_{\ell^2(\ZZ^2)}^2.
    \]
    In particular, taking $\delta = \left(\frac{\pi\lambda_0^2s^2}{4\mu_0L}\right)^{\frac{1}{3}}$, we get 
    \[
    \Vert H_+^{-1} u \Vert_{\ell^2(\ZZ^2)}^2 \leq C_0^2\left(\frac{L}{s^2}\right)^{\frac{2}{3}}\Vert u\Vert_{\ell^2}^2, \qquad C_0 = 2^{\frac{13}{6}}\pi^{\frac{2}{3}}\mu_0^{-\frac{2}{3}}\lambda_0^{-\frac{2}{3}}.
    \]
This completes the proof.     \end{proof}

\begin{proof}[Proof of Proposition \ref{prop: He_invertible}]

(1) We have \begin{equation}
    \label{eq: He_to_E}
\begin{split}
H_e &= \1_{\Omega}H_+\1_{\Omega} + \1_{\Omega^c} (H_+ - 2S)\1_{\Omega^c} + \1_\Omega H_+ \1_{\Omega^c} + \1_{\Omega^c} H_+ \1_{\Omega}\\
&= \1_{\Omega}H_+\1_{\Omega} + \1_{\Omega^c} H_-\1_{\Omega^c} + \1_\Omega H_+ \1_{\Omega^c} + \1_{\Omega^c} H_+ \1_{\Omega}.  
\end{split}
\end{equation}
By \eqref{eq: He_to_E},  
        \[
        E = H_e - \1_{\Omega}H_+\1_{\Omega} + \1_{\Omega^c} H_-\1_{\Omega^c} = \1_\Omega H_+ \1_{\Omega^c} + \1_{\Omega^c} H_+ \1_{\Omega}.
        \]
    Since $\1_\Omega H_+ \1_{\Omega^c}(x,y) = \1_\Omega(x) H_+(x,y) \1_{\Omega^c}(y) = 0$ if $|x - y| > 2$, we have $E$ satisfies \eqref{eq-1a}; thus, $H_e$ satisfies Assumption \ref{a2}.

    (2)
    Recall that since $\lambda(\xi)$ never vanish, $H_+$ is invertible. Thus
    \[
    H_e = H_+ - 2\1_{\Omega^c} S\1_{\Omega^c} \quad \Leftrightarrow \quad H_+^{-1}H_e = \operatorname{Id} - 2H_+^{-1} \1_{\Omega^c} S\1_{\Omega^c}.
    \]
    To show $H_e$ is invertible, it is enough to show $\Vert 2H_+^{-1} \1_{\Omega^c} S\1_{{\Omega^c}}\Vert <1$.  Since $\|S\| \leq 6s$,
        \begin{equation}
    \begin{split}
        \Vert 2H_+^{-1} \1_{{\Omega^c}}S\1_{\Omega^c} u \Vert_{\ell^2} &\leq C_0 L^{1/3} s^{-2/3}  \Vert \1_{{\Omega^c}}S\1_{\Omega^c} u \Vert_{\ell^2}\\
        &\leq 6C_0 L^{\frac{1}{3}} s^{\frac{1}{3}}\Vert u \Vert_{\ell^2}.
    \end{split}
    \end{equation}
    Thus when $s<\rho_0 L^{-1}$, $\rho_0 = 6^{-3}C_0^{-3}$, we have $\Vert 2H_+^{-1} \1_{{\Omega^c}}S\1_{\Omega^c}\Vert <1$.
    Thus $H_e$ is invertible.\end{proof}

    \begin{remark}
    Numerics actually yield the values
    \begin{equation}
        \lambda_0 = 1, \qquad \mu_0 = \dfrac{3}{\pi \sqrt{26}} \simeq 0.18, \qquad C_0 \simeq 31, \qquad s  < 1.5 \cdot 10^{-7} L^{-1}.
    \end{equation}
    That is, if the second-nearest neighbor hopping is much smaller than the first-nearest neighbor hopping (depending on $L$), then a topological insulator fitting in a strip of width $L$ may not have edge spectrum.
\end{remark}


\appendix
\section{proof of some estimates}
\label{App}

\begin{proof}[Proof of \eqref{eq-1e}] Fix $a \in (0,1]$. Then:
\begin{equation}
    \sum_{s \in \ZZ} e^{-2a|s|} = \dfrac{2}{1-e^{-2a}}-1 = \dfrac{1+e^{-2a}}{1-e^{-2a}} = \dfrac{1}{\tanh(a)}\leq \frac{2}{a}
\end{equation}
where the last inequality follows from the fact that $\tanh(x)$ is concave when $x>0$; thus $\tanh(a) \geq \tanh(1) a \geq a /2$ for $a \in (0,1]$. This yields the first inequality in \eqref{eq-1e}.   The second inequality follows immediately since $e^{-2a|\bfx|} = e^{-2a|x_1|} e^{-2a|x_2|}$.

Now fix $r \geq 0$. If $|\bfx| \geq r$, then either $|x_1| \geq r/2$ or $|x_2| \geq r/2$. This induces a splitting in two mutually summetric sums:
\begin{equation}
    \sum_{|\bfx| \geq r} e^{-2a|\bfx|} \leq 2 \sum_{|x_1| \geq r/2, \ x_2 \in \ZZ} e^{-2a|\bfx|} \leq 2\left(\dfrac{2}{a}\right)^2 e^{-a r}.
\end{equation}
This completes the proof. \end{proof}

\begin{proof}[Proof of \eqref{eq-1g}] We apply Schur's test. For a selfadjoint operator, it reads:
\begin{align}
    \| H \| \leq \sup_{\bfx \in \ZZ^2} \sum_{\bfy \in \ZZ^2} \big| H(\bfx,\bfy) \big| 
    & \leq \dfrac{1}{\nu} \sup_{\bfx \in \ZZ^2} \sum_{\bfy \in \ZZ^2} e^{-2\nu|\bfx-\bfy|} 
    \\
    & = \dfrac{1}{\nu} \sum_{\bfy \in \ZZ^2} e^{-2\nu|\bfy|} = \dfrac{1}{\nu} \left(\sum_{y_1 \in \ZZ} e^{-2\nu|y_1|}\right)^2 \leq \dfrac{4}{\nu^3}.
\end{align}
In the last inequality, we used the first inequality in \eqref{eq-1e}, which is valid since $\nu \in (0,1]$.    
\end{proof}

\begin{proof}[Proof of \eqref{eq-1h}] By the slope inequality for convex functions $f(x) = e^{xs}$ with any $s \geq 0$, we have for $0 < \alpha \leq \nu$:
\begin{equation}
    \dfrac{e^{\alpha s} -1}{\alpha} \leq \dfrac{e^{\nu s} -1}{\nu}.
\end{equation}
By applying this inequality to $s = |\bfx-\bfy|$, we deduce that $\frac{S_\alpha}{\alpha} \leq \frac{S_\nu}{\nu}$. We now estimate $S_\nu$. We have 
\begin{align}
    S_\nu = \sup_{\bfx\in \ZZ^2} \sum_{\bfy\in \ZZ^2} |H(\bfx, \bfy)|\left(e^{\nu |\bfx - \bfy|} - 1\right)
    &     \leq \dfrac{1}{\nu} \sup_{\bfx\in \ZZ^2} \sum_{\bfy\in \ZZ^2} e^{-2\nu|\bfx-\bfy|}\left(e^{\nu |\bfx - \bfy|} - 1\right)
    \\
    & \leq \dfrac{1}{\nu} \sum_{\bfy\in \ZZ^2} e^{-\nu|\bfy|} \leq \dfrac{16}{\nu^3},
\end{align}
where we used \eqref{eq-1e} again.
\end{proof}

\begin{proof}[Proof of \eqref{eq-1r}] 1. We first note that we have, by $|t - s| + |s| \geq |t|$ and \eqref{eq-1e}:
\begin{equation}\label{eq-1o}
\sum_{s \in \ZZ} e^{-2\beta|t - s| - 2\beta |s|} \leq e^{-\beta|t|} \sum_{s\in \ZZ} e^{-\beta|t - s| - \beta|s|} \leq e^{-\beta|t|}\sum_{s\in \ZZ} e^{- \beta|s|} \leq \frac{4e^{-\beta|t|}}{\beta}.
\end{equation}

2. We now control $S$, the sum in the LHS of \eqref{eq-1r}. To this end, we apply \eqref{eq-1o} four times: first to $(t,s) = (x_1,y_1)$, then $(w_1,z_1)$, then $(y_2,z_2)$ and finally $(x_2,y_2)$. This gives:
\begin{align}
    S & \leq \dfrac{4}{\beta} e^{-\beta|x_1|- 2\beta |w_2|} \sum_{y_2,\bfz} e^{-2\beta |x_2-y_2| - 2\beta |y_2-z_2| -2\beta |z_1| - 2\beta |z_1-w_1| - 2\beta |z_2| }
    \\
    & \leq \left( \dfrac{4}{\beta} \right)^2 e^{-\beta|x_1|-\beta|w_1| - \beta|w_2|} \sum_{y_2,z_2} e^{-2\beta |x_2-y_2| - 2\beta |y_2-z_2| - 2\beta |z_2|}
    \\
    & \leq \left( \dfrac{4}{\beta} \right)^3 e^{-\beta|x_1|-\beta|\bfw|} \sum_{y_2} e^{-2\beta |x_2-y_2| - 2\beta |y_2|} \leq \left( \dfrac{4}{\beta} \right)^4 e^{-\beta|\bfx|-\beta|\bfw|}.
\end{align}
This is \eqref{eq-1r}. 
\end{proof}

\bibliographystyle{amsxport}
\bibliography{ref.bib}

\begin{bibdiv}
\begin{biblist}

\bib{ASS94}{article}{
      author={Avron, Joseph},
      author={Seiler, Ruedi},
      author={Simon, Barry},
       title={Charge deficiency, charge transport and comparison of
  dimensions},
        date={1994},
     journal={Communications in mathematical physics},
      volume={159},
       pages={399\ndash 422},
}

\bib{ASV13}{article}{
      author={Avila, Julio~Cesar},
      author={Schulz{-}Baldes, Hermann},
      author={Villegas{-}Blas, Carlos},
       title={Topological invariants of edge states for periodic
  two-dimensional models},
        date={2013},
     journal={Mathematical Physics, Analysis and Geometry},
      volume={16},
      number={2},
       pages={137\ndash 170},
}

\bib{AW15}{book}{
      author={Aizenman, Michael},
      author={Warzel, Simone},
       title={Random operators},
   publisher={American Mathematical Soc.},
        date={2015},
      volume={168},
}

\bib{B19}{article}{
      author={Bal, Guillaume},
       title={Continuous bulk and interface description of topological
  insulators},
        date={2019},
        ISSN={0022-2488,1089-7658},
     journal={J. Math. Phys.},
      volume={60},
      number={8},
       pages={081506, 20},
         url={https://doi.org/10.1063/1.5086312},
      review={\MR{3993755}},
}

\bib{bal2022semiclassical}{article}{
      author={Bal, Guillaume},
       title={Semiclassical propagation along curved domain walls},
        date={2022},
      eprint={arXiv:2206.09439},
}

\bib{bal2021edge}{article}{
      author={Bal, Guillaume},
      author={Becker, Simon},
      author={Drouot, Alexis},
      author={Kammerer, Clotilde~Fermanian},
      author={Lu, Jianfeng},
      author={Watson, Alexander},
       title={Edge state dynamics along curved interfaces},
        date={2021},
      eprint={arXiv:2106.00729},
}

\bib{BH13}{article}{
      author={Bernevig, Andrei},
      author={Hughes, Taylor},
       title={Topological insulators and topological superconductors},
        date={2013},
}

\bib{CT73}{article}{
      author={Combes, Jean-Michel},
      author={Thomas, L.},
       title={Asymptotic behaviour of eigenfunctions for multiparticle
  {S}chr{\"o}dinger operators},
        date={1973},
}

\bib{D21}{article}{
      author={Drouot, Alexis},
       title={Microlocal analysis of the bulk-edge correspondence},
        date={2021},
     journal={Communications in Mathematical Physics},
      volume={383},
       pages={2069\ndash 2112},
}

\bib{D22}{article}{
      author={Drouot, Alexis},
       title={Topological insulators in semiclassical regime},
        date={2022},
      eprint={arXiv:2206.08238},
}

\bib{EG02}{article}{
      author={Elbau, P.},
      author={Graf, G.~M.},
       title={Equality of bulk and edge {H}all conductance revisited},
        date={2002},
        ISSN={0010-3616,1432-0916},
     journal={Comm. Math. Phys.},
      volume={229},
      number={3},
       pages={415\ndash 432},
         url={https://doi.org/10.1007/s00220-002-0698-z},
      review={\MR{1924362}},
}

\bib{EGS05}{article}{
      author={Elgart, A.},
      author={Graf, G.~M.},
      author={Schenker, J.},
       title={Equality of the bulk and edge {H}all conductances in a mobility
  gap},
        date={2005},
     journal={Communications in mathematical physics},
      volume={259},
       pages={185\ndash 221},
}

\bib{FGW00}{article}{
      author={Fr{\"o}hlich, Jurg},
      author={Graf, Gian~Michele},
      author={Walcher, J.},
       title={On the extended nature of edge states of quantum hall
  hamiltonians},
organization={Springer},
        date={2000},
      volume={1},
       pages={405\ndash 442},
}

\bib{GP13}{article}{
      author={Graf, Gian~Michele},
      author={Porta, Marcello},
       title={Bulk-edge correspondence for two-dimensional topological
  insulators},
        date={2013},
     journal={Communications in Mathematical Physics},
      volume={324},
       pages={851\ndash 895},
}

\bib{GT18}{article}{
      author={Graf, Gian~Michele},
      author={Tauber, Cl{\'e}ment},
       title={Bulk--edge correspondence for two-dimensional floquet topological
  insulators},
        date={2018},
     journal={Ann. Henri Poincar{\'e}},
      volume={19},
       pages={709\ndash 741},
}

\bib{GT20}{article}{
      author={Graf, Gian~Michele},
      author={Thiang, Guo~Chuan},
       title={End of talk "coarse index and spectral flow methods for
  topological matter", https://vimeo.com/showcase/7885383/video/490849974},
        date={2020},
}

\bib{H88}{article}{
      author={Haldane, Duncan},
       title={Model for a quantum hall effect without landau levels:
  Condensed-matter realization of the" parity anomaly"},
        date={1988},
     journal={Physical review letters},
      volume={61},
      number={18},
       pages={2015},
}

\bib{PPY22}{article}{
      author={Hu, Pipi},
      author={Xie, Peng},
      author={Zhu, Yi},
       title={Traveling edge states in massive dirac equations along slowly
  varying edges},
        date={2023},
     journal={IMA Journal of Applied Mathematics},
      volume={88},
      number={3},
       pages={455\ndash 471},
}

\bib{KS02}{article}{
      author={Kellendonk, J.},
      author={Richter, T.},
      author={Schulz-Baldes, H.},
       title={Edge current channels and {C}hern numbers in the integer quantum
  {H}all effect},
        date={2002},
        ISSN={0129-055X,1793-6659},
     journal={Rev. Math. Phys.},
      volume={14},
      number={1},
       pages={87\ndash 119},
         url={https://doi.org/10.1142/S0129055X02001107},
      review={\MR{1877916}},
}

\bib{LT22}{article}{
      author={Ludewig, Matthias},
      author={Thiang, Guo~Chuan},
       title={Cobordism invariance of topological edge-following states},
        date={2022},
     journal={Advances in Theoretical and Mathematical Physics},
      volume={26},
      number={3},
       pages={673\ndash 710},
}

\bib{O22}{article}{
      author={Ojito, D.~P.},
       title={Interface currents and corner states in magnetic quarter-plane
  systems},
        date={2022},
     journal={Preprint, arXiv:2212.06234},
}

\bib{S15}{article}{
      author={Shapiro, Jacob},
       title={Notes on topological insulators and mathematical physics,
  \url{https://web.math.princeton.edu/~js129/PDFs/Top_SSP_Lecture_Notes.pdf }},
        date={2015},
}

\bib{ST19}{article}{
      author={Shapiro, Jacob},
      author={Tauber, Cl\'{e}ment},
       title={Strongly disordered {F}loquet topological systems},
        date={2019},
        ISSN={1424-0637,1424-0661},
     journal={Ann. Henri Poincar\'{e}},
      volume={20},
      number={6},
       pages={1837\ndash 1875},
         url={https://doi.org/10.1007/s00023-019-00794-3},
      review={\MR{3956162}},
}

\bib{T20}{article}{
      author={Thiang, Guo~Chuan},
       title={Edge-following topological states},
        date={2020},
     journal={J. Geom. Phys.},
      volume={156},
       pages={103796, 17},
}

\end{biblist}
\end{bibdiv}
\end{document}